\RequirePackage{fix-cm}
\documentclass[smallextended]{svjour3}       
\smartqed  
\usepackage[utf8]{inputenc}
\usepackage[english]{babel}
\usepackage{amsmath}
\usepackage{dsfont}
\usepackage{mathrsfs}
\usepackage{amsfonts}
\usepackage{amssymb}
\usepackage{makeidx}
\usepackage{graphicx}
\usepackage{epstopdf}
\usepackage[left=4cm,right=3cm,top=3cm,bottom=3cm]{geometry}
\usepackage{float}
\usepackage{cite}
\usepackage{ragged2e}
\usepackage{csquotes}
\usepackage{color}

\usepackage{mathtools}

\newtheorem{mydef}{Definition}
\newtheorem{maintheorem}{Main Theorem}
\newtheorem{col}{Corollary}
\newtheorem{con}{Conjecture}
\newtheorem{observation}{Observation}

\newcommand{\NN}{{\mathbb N}}

\begin{document}

\title{Ancestral sequence reconstruction with Maximum Parsimony}


\author{Lina Herbst         \and
        Mareike Fischer 
}


\institute{L. Herbst \at
              Institute for Mathematics and Computer Science,
Greifswald University, 
Walther-Rathenau-Str. 47, 17489 Greifswald, Germany \\
              \email{herbstl@uni-greifswald.de}           
           \and
           M. Fischer \at
              Institute for Mathematics and Computer Science,
Greifswald University, 
Walther-Rathenau-Str. 47, 17489 Greifswald, Germany\\
              \email{email@mareikefischer.de} 
}

\date{Received: date / Accepted: date}

\maketitle

\begin{abstract} One of the main aims in phylogenetics is the estimation of ancestral sequences based on present-day data like, for instance, DNA alignments. One way to estimate the data of the last common ancestor of a given set of species is to first reconstruct a phylogenetic tree with some tree inference method and then to use some method of ancestral state inference based on that tree. One of the best-known methods both for tree inference as well as for ancestral sequence inference is Maximum Parsimony (MP). In this manuscript, we focus on this method and on ancestral state inference for fully bifurcating trees. In particular, we investigate a conjecture published by Charleston and Steel in 1995 concerning the number of species which need to have a particular state, say $a$, at a particular site in order for MP to unambiguously return $a$ as an estimate for the state of the last common ancestor. We prove the conjecture for all even numbers of character states, which is the most relevant case in biology. We also show that the conjecture does not hold in general for odd numbers of character states, but also present some positive results for this case.

\noindent

\end{abstract}

\section{Introduction}
\label{intro}
Reconstructing data, like e.g. DNA sequences, of common ancestors of species living today is one of the main challenges of modern phylogenetics. Ancestral sequence reconstruction is a growing research field as information about the genetic makeup of ancestors is important in various areas, for instance in drug design, protein function investigation and comparative genomics \cite{liberles, moretaxa, protein}. Moreover, ancestral sequence estimation is crucial for understanding the evolution of today's species. 

One possible way to estimate ancestral sequences is to consider the underlying evolutionary  tree of the species of interest and to use Maximum Parsimony (MP) to infer ancestral sequences based on this tree. Some reasons that make MP interesting is that it does not require estimates of branch lengths on the input tree and that as a purely combinatorial method, it is a very fast method of ancestral sequence reconstruction \cite{fitch,phylogenetics}. In this manuscript, we are interested in the last common ancestor of all species under investigation, not the ancestors of any subgroups. This means that we can use the famous Fitch algorithm \cite{fitch} to infer {\it all }  possible ancestral states (we will explain this more in-depth in the following section).
For now just recall that the Fitch algorithm goes from the leaves to the root in a binary tree, i.e. in a tree where all nodes have exactly two children, and assigns each parental node a state set based on the state sets of its children. If the children's state sets intersect, this intersection is taken as the estimated state set for the ancestor. Otherwise, the union of the two sets is taken. For instance, in Figure \ref{cat}, coming from the lowermost leaf, most inner nodes will be assigned $\{b\}$, but the parent of leaf 2 will be assigned $\{a,b\}$, because set $\{a\}$ coming which from leaf 2 and set $\{b\}$ coming from leaves $3, \ldots,n$ (where $n$ is the number of taxa) do not intersect. Then, root $\rho$ is assigned $\{a\}$ as the intersection of set $\{a\}$ coming from leaf 1 and set $\{a,b\}$ coming from leaves $2, \ldots, n$, intersect. This means that the root of the tree will unambiguously be assigned $a$ in this example.

However, MP cannot always make unambiguous decisions about the root sequence. Sometimes, for a particular site, MP might be indecisive between, say, the two DNA nucleotides $A$ and $C$ (note that our study is not restricted to the DNA alphabet, but can rather be applied to all kinds of alphabets, which is why we use lowercase letters in the following). Of course, for ancestral sequence reconstruction the question when the assignment is unique, i.e. when the decision is unambiguous, is of the utmost importance. In this regard, the following question arises: How many present-day species, i.e. how many leaves of the tree, need to be in a certain state, say $a$, such that the Maximum Parsimony estimate of the last common ancestor is also $a$? Mathematically speaking, this question is equivalent to the question how many leaves need to be assigned state $a$ (or `colored` with `color' $a$), such that $a$ is the unambiguous root state estimate in the underlying phylogenetic tree. 

It can be easily seen that the answer to this question depends on the tree shape -- for instance, the number of taxa; but also the height of the tree plays a crucial role. For instance, consider again the so-called caterpillar tree, which is depicted in Figure \ref{cat}. This is an extreme case: As we have explained above, coloring two leaves, namely the uppermost two, with $a$ is always sufficient to give an unambiguous root state estimate of $a$, even if all other leaves, which may be the vast majority, are in a different state. So for the caterpillar, two leaves in state $a$ are always sufficient to give an unambiguous $a$ root state estimate -- regardless of the number of taxa under consideration. 

\begin{figure}[H]
\centering
\includegraphics[scale=0.35]{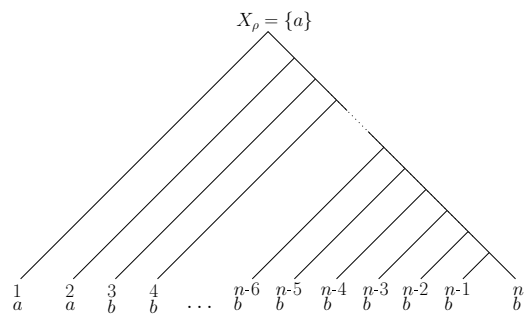}
\caption{A caterpillar tree, whose last common root is assigned $\{a\}$ by just assigning two leaves $a$. Note that this construction works even if the total number of taxa in the tree is large.}
\label{cat}
\end{figure}

On the other hand, if the number of species is $n=2^k$ for some natural number $k$, then you can also consider the so-called fully resolved and balanced tree, which is also often referred to as the fully bifurcating tree. The case $k=3$, i.e. $n=8$, is depicted in Figure \ref{treeh3}. This tree turns out to be the opposite extreme: As all leaves are separated from the root by the same number of edges, it is clear that no constant number of leaves can have such a strong impact as in the caterpillar case, but that instead the number of leaves which need to be assigned $a$ in order for MP to return $a$ as the unambiguous root state estimate will grow as $k$ (and thus $n$) grows. So in this regard, the fully bifurcating tree is of the utmost importance -- once you know the minimal number of leaves that need to be assigned $a$ in this particular tree in order for Maximum Parsimony to unambiguously return $a$ as the root state estimate, you can use this number as an upper bound for {\it any} given tree on $n=2^k$ taxa.

So the present manuscript is concerned with the problem of finding the minimal number of leaves that need to be assigned $a$ such that the MP root state estimate is $\{a\}$ when the underlying tree is fully bifurcating. This problem has first been addressed in \cite{paper}, where the authors stated a recursive formula which used the Fibonacci numbers and which explicitly calculates the desired number in case the alphabet under consideration consists only of $r=2$ states. So biologically, this result would only be applicable to binary data like for instance transversions versus transitions, but not to data like RNA or DNA (four states), proteins (20 states) or codons (64 states). However, in their manuscript the authors stated a conjecture for general alphabets of arbitrary size $r$. This conjecture inspired our work. In fact, in the present manuscript we first prove the conjecture for the case where $r$ is even (which makes the result applicable to biological data like DNA, RNA, proteins or codons as stated above). We then analyze the more involved case where $r$ is odd and show why the conjecture does not generally hold in this case, but also show some positive results for this case, like a simplification to only two different cases that need to be considered when $r=3$.

\subsection{Preliminaries}
\subsubsection{Basic definitions and notation}

Before we can present our results, we first have to introduce some basic concepts. We start with some definitions. Recall that a {\it rooted binary phylogenetic tree on a taxon set $L$} (also sometimes referred to as a {\it rooted binary phylogenetic $L$-tree}) is a connected, acyclic graph in which the vertices of degree at most 1 are called {\it leaves}, and in which there is exactly one node $\rho$ of degree 2, which is referred to as {\it root}, and all other non-leaf nodes have degree 3. Moreover, in a rooted binary phylogenetic $L$-tree the leaves are bijectively labelled by the elements of $L$. Note that in the special case where there is only one node, this node can be regarded as a leaf {\it and} as a root at the same time, i.e. this is the only case where the root does not have degree 2. 

Throughout this paper, when we refer to trees, we always mean rooted binary phylogenetic $L$-trees, and we assume without loss of generality that $L=\{1, \dots n\}$. Note all nodes on the path from a node $v$ to the root $\rho$ (except for $v$) are called {\it ancestors} of $v$, and $v$ is a {\it descendant} of all these nodes. Whenever the path from $v$ to an ancestor $u$ consists of only one edge, $u$ is called {\it direct ancestor} of $v$, and $v$ is called {\it direct descendant} of $u$.

Now let $k \in \NN$. Then a \textit{fully bifurcating phylogenetic tree of height $k$, $T_k$,} is a rooted binary phylogenetic tree which has exactly $n=2^k$ leaves and height $k$. Here,  {\it height} refers to the number of edges between the root and each leaf (note that this number is uniquely determined for rooted binary trees with $n=2^k$ leaves). Figure \ref{treeh3} depicts the case $n=2^3=8$, i.e. tree $T_3$. Note that for $k=0$, i.e. $n=2^0=1$, the corresponding fully bifurcating phylogenetic tree $T_0$ consists of only one node, which is at the same time its root and only leaf.

\begin{figure}
\centering
\includegraphics[scale=0.5]{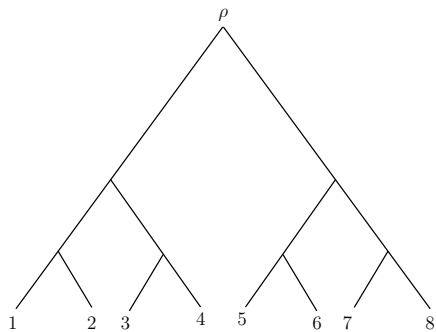}
\caption{$T_3$ is a fully bifurcating phylogenetic tree of height 3 with $2^3=8$ leaves.}
\label{treeh3}
\end{figure}

Now that we introduced tree $T_k$, which we will investigate in this manuscript, we need to introduce two more concepts: the kind of data we will map onto the leaves of the tree as well as the method with which we will then estimate ancestral data like in particular the data of the root of the tree.

We first consider the data. Assume we have an alphabet $R$ with $|R|=r$ so-called {\it character states}, like e.g. the DNA alphabet $R=\{A,C,G,T\}$, where $r=4$. Then a {\it character} on a taxon set $L$ is just a function $f$ from $L$ to $R$, i.e. $f : L \rightarrow R.$ 
If $| f(L) | = \tilde{r}$, then $f$ is called $\tilde{r}$-state character. Note that as we consider $L=\{1,2, \ldots, n\}$, we often write $f=f(1)f(2)\ldots f(n)$ instead of listing $f(1)$, $f(2)$, etc. explicitly. For instance, Figure \ref{cat} shows an example of character $f=aab\dots b$ mapped onto the leaves of a rooted binary phylogenetic tree. Note that given real data like e.g. a DNA alignment, a character can just be regarded as a column in the alignment, because each species is assigned one character state (however, note that in our case, alignments are considered to be gap-free). In this context, biologists also often refer to characters as {\it sites}. 

We now want to turn our attention to the method which we will use to estimate ancestral sequence data, i.e. sequences of ancestral states. Given a rooted binary phylogenetic tree $T$, the \textit{standard decomposition} of $T$ is the decomposition of $T$ into its two maximal rooted pending subtrees \cite{combinatorial}. Informally speaking, as the root $\rho$ has degree 2, this means that the standard decomposition considers the two subtrees directly below $\rho$. An illustration of this decomposition is given by Figure \ref{decomp}. Note that the fully bifurcating tree $T_k$, with which this manuscript is concerned, can be decomposed into two trees of shape $T_{k-1}$, which in turn can be decomposed into two $T_{k-2}$ trees and so forth.

\begin{figure}[H]
\centering
\includegraphics[scale=0.45]{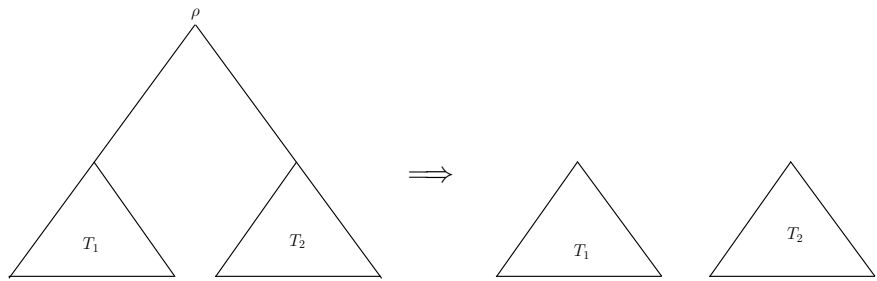}
\caption{The standard decomposition of a rooted binary phylogenetic tree. The tree is decomposed into its two maximal rooted subtrees $T_1$ and $T_2$.}
\label{decomp}
\end{figure}

A frequently used method for ancestral state reconstruction is {\it Maximum Parsimony}, or {\it MP} for short \cite[Section 5]{phylogenetics}. Given a tree and a character, MP seeks to find possible state assignments to all internal nodes such that the number of edges in the so-called {\it change set} is minimized. The change set consists of all edges of $T$ whose two endpoints are assigned different states. The minimal number of edges in the change set for a given tree $T$ and a given character $f$ is called {\it Maximum Parsimony score} (or {\it MP score} for short) of $f$ on $T$. The MP score of an entire alignment, i.e. of a sequence of characters on $T$, is then simply the sum of all scores induced by the individual characters. This is why we investigate only individual characters in this manuscript rather than entire alignments.  

Calculating the MP score for a given tree and a given character can be done in polynomial time using the famous {\it Fitch algorithm} \cite{fitch}, which in turn is based on {\it Fitch's parsimony operation}, which we define next.

\begin{mydef}{(Fitch's parsimony operation)}\label{parsimonyoperation}\\
Let $R$ be a nonempty finite set and let $A, B \subseteq R$. Then, \textit{Fitch's parsimony operation} $*$ is defined by 
$$A*B \coloneqq \begin{cases}
A \cap B, & \text{if } A \cap B \neq \emptyset, \\
A \cup B, & \text{otherwise.}
\end{cases}$$
This means that whenever the intersection of sets $A$ and $B$ is not empty, $A*B$ is set to be equal to this intersection. Otherwise, the union is taken instead of the intersection.
\end{mydef}
Using this operation, the {\it Fitch algorithm} \cite{fitch} works as follows (note that what we call the Fitch algorithm is in fact only one phase of the algorithm described in \cite{fitch}, but this is the only part required to estimate possible root states and is thus sufficient for the present manuscript). 
Assume we have a rooted binary tree $T$ with leaf set $L$ and a character $f: L\rightarrow R$. Now the vertices of $T$ are recursively assigned non-empty subsets of $R$. First, each leaf is assigned the set consisting of the state assigned to it by $f$. Then all other vertices $v$, whose two direct descendants have already been assigned a subset, say $A$ and $B$, are assigned the set $A*B$. This step is continued upwards along the tree until the root $\rho$ is assigned a subset, which is denoted by $X_\rho$. 
The sets assigned to each node represent possible state sets which lead to assignments of states to all internal nodes, so-called {\it extensions of character $f$}, such that the number of edges in the change set is minimized and thus the parsimony score is realized on $T$ given $f$. Note that Fitch described two more phases \cite{fitch}. One of them only  derives extensions once the state sets for each node are given. However, the other one adds more states to some internal node sets as the procedure described above might not recover all possible states for all internal nodes \cite{felsenstein}. However, note that the states suggested for the root $\rho$ by the described algorithm do not require any correction \cite{fitch}. Therefore, as the present manuscript is mainly concerned with $\rho$, the described procedure is sufficient and no corrections need to be made to $X_\rho$. 

We are now in the position to turn our attention to the main question of our manuscript: Given a fully bifurcating phylogenetic tree $T_k$ on taxon set $L$ and height $k$, as well as an alphabet $R$ with $|R|=r$, some set $A\subseteq R$ and $a \in A$,  what is the minimal number of leaves which must be assigned $a$ such that $X_\rho=A$? Given $r \geq 2$ and $| A | \leq 2^k$, we denote this minimum by  $f_{k,r}^A$. For simplicity, let $f_{k,r} := f_{k,r}^{\{a\}}$. 

With the standard decomposition of rooted binary trees and Fitch's algorithm, it can be easily seen that 
\begin{equation}
f_{k+1,r}^A = \min_{\substack{B,C \subseteq R, \\ |B|,|C| \leq 2^k}} \{f_{k,r}^B+f_{k,r}^C: B*C=A \}, 
 \label{formel}
 \end{equation}
where $*$ is the parsimony operation.

Note that the minimum is taken over all possible subsets $B$ and $C$ of $R$, which could be assigned to the direct descendants of $\rho$ and which would then be used by Fitch's operation to calculate $X_\rho$. Whenever there is no ambiguity, we omit the subscripts and only refer to $\min \{f_{k,r}^B+f_{k,r}^C\}$. The main aim of this manuscript is to find a way to calculate this minimum without having to consider all possible sets $B$ and $C$. Our work was motivated by the fact that in \cite{paper} it was shown by M. Steel and M. Charleston that $f_{k,2}$ equals the $(k+1)$th Fibonacci number; so $f_{k,2}$ can be calculated without considering all possible state assignments to the nodes of $T_{k-1}$. However, the authors did not consider the general case $r>2$, but instead stated the following conjecture for this case which we seek to answer in the present manuscript.

\begin{con}[Charleston and Steel \cite{paper}]
\label{conjecture} $\mbox{ \hspace{8cm}}$
For $T_k$ and $r \geq 2$ we have:
\begin{align} f_{k,r}= 
\begin{cases}
f_{k-p,r}+f_{k-p-1,r} & \text{when } r=2p, \mbox{\hspace{0.1cm}} p \in \NN_{\geq 1},\\
2 \cdot f_{k-p,r} & \text{when } r=2p-1, \mbox{\hspace{0.1cm}} p \in \NN_{\geq 2}
\end{cases}.
\end{align}
\end{con}

\par\vspace{0.5cm} Note that Conjecture \ref{conjecture} is only defined for $k \geq p+1$ if $r=2p$ and $k \geq p$ if $r=2p-1$, because $f_{k,r}$ is only defined for $k\geq 0$ as $k$ is the height of the underlying tree $T_k$. Also note that $f^A_{k,r}$ is only defined whenever $|A| \leq 2^k$, because by the Fitch operation you can never get more states into the root estimate set than there are leaves (as each leaf might be assigned a different state, which would lead to a root state set of size $2^k$, but more than that is not possible). So in the following, even if we do not state explicitly that our theorems are only true for the regions where $f_{k,r}^A$ is defined, we will implicitly always assume that only these regions are considered. In particular, our proofs only consider these cases.

It is important to note that the initial conditions required for the start of the recursion are not specified in Conjecture \ref{conjecture}. In the following we show that Conjecture \ref{conjecture} is not valid in general, because the case where $r$ is odd is problematic unless the initial conditions are chosen in a specific way, and we also show relevant examples to illustrate this result.

Before we present our results, we want to introduce one last piece of notation, which simplifies all following equations a lot: In the following, $A_i$ always denotes a subset of our alphabet $R$ that contains $i$ elements in total, one of which is $a$, i.e. $A_i \subseteq R$ with $a \in A_i$ and $|A_i| =i$. So in particular, set $A_i$ contains $a$ for all $i$. On the other hand, let $D_j \subseteq R \setminus A_i$ with $|D_j| =j$ and $j=1, \dots, r-i$. In particular, $D_j$ does not contain $a$ for any $j$, and $i+j=r=|R|$.

\section{Results} \label{section2}
\subsection{First insights into the case where $r$ is odd}

Before we present our main result, which is concerned with the case in which the size $r$ of the alphabet under consideration is even, we first consider the case where $r$ is odd and present an example for which Conjecture \ref{conjecture} fails when the initial conditions are not chosen in a particular way. 

\begin{observation} \label{obs} Let $p \in \NN_{\geq 2}$ and $r=2p -1$. Then for $k=p+1$, Conjecture \ref{conjecture} suggests $f_{p+1,r} = 2 \cdot f_{(p+1)-p,r} = 2 \cdot f_{1,r} = 2 \cdot 2 = 4$. This is due to the fact that in order for the root $\rho$ of $T_1$, i.e. the tree that only consists of $\rho$ and two leaves, to be assigned ${a}$, both leaves need to be in state $a$, so $f_{1,r} = 2$ (this scenario is depicted on the left-hand side of Figure \ref{rschlange}). However, $f_{p+1,r}=4$ is incorrect. In fact, $f_{p+1,r}=3$.
\end{observation}

We will now prove Observation \ref{obs} and illustrate the construction presented in the proof subsequently with the explicit example where $r=3$.

\begin{proof} \mbox{}
Our proof strategy is depicted in Figure \ref{proofidea} and works as follows:
\begin{figure}[H]
\centering
\includegraphics[scale=0.6]{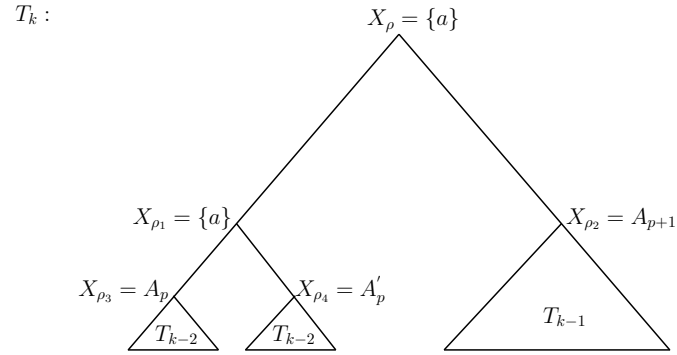}
\caption{This figure illustrates the idea of the proof: We consider $T_k$ and its standard decomposition into two trees $T_{k-1}$, for one of which we consider its standard decomposition of two trees $T_{k-2}$.}
\label{proofidea}
\end{figure}
We consider tree $T_k$ and its standard decomposition into two trees $T_{k-1}$. We will construct a character such that one of the two $T_{k-1}$ trees is assigned $A_{p+1}$ as MP root estimate, i.e. $p+1$ distinct character states, one of which is $a$. For the other subtree $T_{k-1}$, we again consider its standard decomposition into two subtrees $T_{k-2}$, both of which we will assign sets $A_p$ and $A_p'$, respectively, i.e. each subtree $T_{k-2}$ gets assigned a set with $p$ elements (one of which is $a$) to its root. The trick is that as we have $r=2p-1$ states, we can achieve the latter by letting $A_p=\{a,c_1,\ldots,c_{p-1}\}$ and $A_p'=\{a,c_p,\ldots,c_{2p-2}\}$, where $R=\{a,c_1,\cdots,c_{2p-2}\}$ and thus $|R|=r=2p-1$, and so $A_p \cap A_p' = \{a\}$. Thus, the root of $T_{k-1}$ would have set $\{a\}$ as a root estimate, and $T_k$ would have the root estimate $\{a\} * A_{p+1}=\{a\} \cap A_{p+1}=\{a\}$ as required. We will present a construction that can achieve this by assigning $a$ to only three leaves in total, one in one of the $T_{k-1}$ subtrees and one in each $T_{k-2}$ subtree of the other $T_{k-1}$ subtree as described above. 

\par \vspace{0.5cm}
So we first consider tree $T_{k-1}=T_p$ and consider its subtrees as depicted in Figure \ref{paperfig}. For this tree, we construct a character that only employs one $a$ and leads to an MP root assignment of $A_{p+1}$, i.e. of a set with $p+1$ elements, one of which is $a$. The construction is as follows: We assign leaf $v$ with $a$ and all leaves in subtrees $T_0, \ldots, T_{p-1}$ are all assigned only one of the $r=2p-1$ character states, but each tree gets a unique one. So in total, the character uses $p+1$ distinct character states, say $a,c_1,\ldots, c_{p}$, and it can easily be seen that the Fitch algorithm will lead to the root assignment $X_\rho = \{a,c_1,\ldots, c_{p}\}=A_{p+1}$. \\ \par
Note that we can use the same construction to assign sets $A_p$ and $A_{p}'$ to trees $T_{k-2}$, where we can choose these two sets such that they intersect only in $a$ as explained above. 

So in total, we have used $a$ three times and derive $X_\rho=\{a\}$. This already contradicts $f_{p+1,r}=4$ as suggested by Conjecture \ref{conjecture}. Moreover, one can show that 3 is best possible, i.e. there is no character assigning $a$ to only two leaves and giving $X_\rho=\{a\}$ (calculation not shown). This completes the proof. \qed
\end{proof}

\begin{figure}[H]
\centering
\includegraphics[scale=0.3]{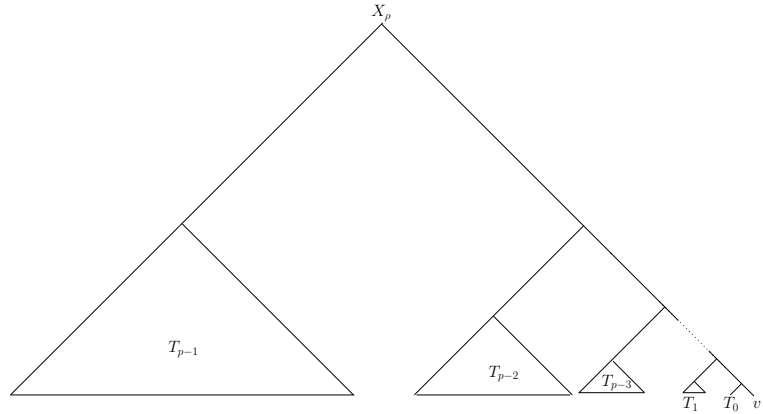}
\caption{The decomposition of the fully bifurcating tree $T_p$ of height $p$ as needed in the construction presented in the proof of Observation \ref{obs} and later in the proof of Lemma \ref{lemmar2p}. In this construction, leaf $v$ is assigned $a$, $T_0$ is assigned another state, and all other subtrees which are sketched here as triangles are assigned one state each such that all states used are pairwise distinct.}
\label{paperfig}
\end{figure}

Note that the problem for $r=2p-1 \geq 3$ as described in the construction of our proof of  Observation \ref{obs} could be avoided if the recursion was started with $f_{p,r}=2$ and $f_{p+1,r}=3$ as initial conditions (rather than $f_{1,r}=2$ and $f_{2,r}=3$). So a careful choice of initial conditions is definitely important when $r$ is odd, and for trees whose height is too small, the recursion is thus not applicable. However, note that the construction we used is not the only way to let Conjecture \ref{conjecture} fail when $r$ is odd -- in fact, there are also other examples, where the conjecture fails (result not shown). This is why in Section \ref{3-state} we will analyze the case $r=3$ more in-depth and provide some simplifications concerning $f_{k,3}$.

We complete this section by providing an example in order to illustrate the construction presented in the proof of Observation \ref{obs}.

\begin{example}
Let $R=\{ a,b,c \}$, i.e. $r=3$. Let $p=2$, and thus $r=2p-1$ as required. By Conjecture \ref{conjecture} we would have for all $k \geq 2$ that $f_{k,3}= 2 \cdot f_{k-2,3}$. Hence $f_{3,3}=2 \cdot f_{3-2,3} = 2 \cdot f_{1,3} = 2 \cdot 2 =4$, since $f_{1,r}=2$ for all $r$ (see Figure \ref{rschlange}). \\
Consider $T_3$ with the character $f=abcabbca$ as depicted in Figure \ref{t_3}. This assignment of states to the leaves shows that $f_{3,3} \leq 3$. Note that an exhaustive search through all possible characters even confirms that $f_{3,3}=3$ (calculation not shown). In any case, this shows that Conjecture \ref{conjecture}, which would lead to $f_{3,3}=4$, is not correct.
\begin{figure}
\centering
\includegraphics[scale=0.6]{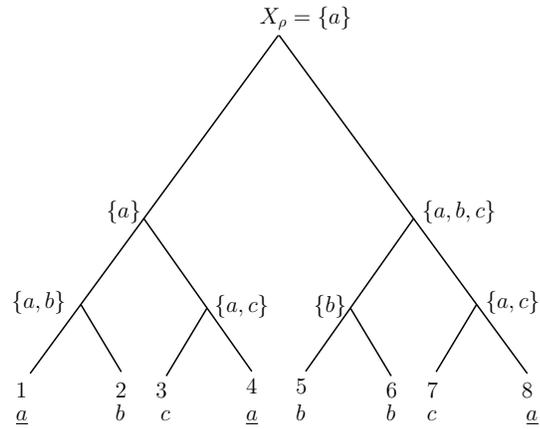} 
\caption{$f_{3,3} \leq f_{2,3}^{\{a\}}+f_{2,3}^{\{a,b,c\}}=2+1=3$, since two leaves have to be assigned $a$ to obtain $\{a\}$ for the root of the left subtree of height 2 and one leaf have to be assigned $a$ to obtain $\{a,b,c\}$ for the root of the right subtree of height 2. Applying the parsimony operation for these two sets results in $X_\rho=\{a\}$.}
\label{t_3}
\end{figure}

\end{example}

So it is already obvious that calculating $f_{k,r}$ recursively is not trivial and that Conjecture \ref{conjecture} does not in general hold when $r$ is odd. However, the present manuscript has two main aims: We want to prove that the conjecture holds for the case where $r$ is even, and we want to prove a recursion for the case where $r=3$. Concerning the second aim, we are now already in the position to state the desired result, which we will prove  subsequently in Section \ref{3-state}. 

\begin{maintheorem}  \label{maintheo3}
Let $T_k$ be a fully bifurcating tree of height $k$ and let $R$ be a set of character states with $|R|=r=3$. Then for all $k \geq 3$ we have
\begin{equation*}
f_{k+1,3} = 2 \cdot f_{k,3}^{A_2} = 2 \cdot f_{k-1,3}.
\end{equation*}
Moreover $f_{2,3}=2$ and $f_{3,3}=3$.
\end{maintheorem}

Note that this theorem basically states that Conjecture \ref{conjecture} is correct for $r=3$ as long as the initial conditions are chosen in the way stated in the theorem. So even if we cannot make a general statement about the involved case where $r$ is odd, this theorem characterizes the cases for which the conjecture holds when $r=3$. 

The proof of this theorem requires some general properties of $f_{r,k}^A$, which will later also be needed to address the case where $r$ is even. So now we first turn our attention to some general properties of $f_{r,k}^A$, some of which will be proved in the appendix.

\subsection{General properties of $f_{k,r}^A$}
We now investigate some general properties of $f_{k,r}^A$. These properties will be used later on to prove our main theorems. However, as we will point out, some properties of $f_{k,r}^A$ are rather intuitive, even if they come with technical proofs. These proofs can be found in the appendix. 

Our main aim in this section is to simplify Equation \eqref{formel} in the following sense: Rather than taking the minimum of all subsets $B$ and $C$ such that $B,C\subseteq R$ and $|B|, |C| \leq 2^k$, we want to see if some sets can already be discarded as they can never lead to the minimum. 

Some first properties of $f_{k,r}^A$ are stated in the following lemma and can be derived from \eqref{formel}.

\begin{lemma}\label{lemma1}
Let $A \subseteq R$, $r=|R|$ and $f_{k,r}^A$ be as defined in \eqref{formel}. Then, the following statements hold: 
\begin{itemize}
\item[(i)] If $a \in A \cap B$ and $|A| = |B|$, then 
$f_{k,r}^A=f_{k,r}^B$.
\item[(ii)] If $k=0$, then $f_{0,r}^{\{a\}}=f_{0,r}=1$ for all $r \in \mathbb{N}$.
\item[(iii)] If $k \geq 1$, then $f_{k,r}^{\{a\}}=f_{k,r} \geq 2$ for all $r \in \mathbb{N}$. 
\end{itemize}
\end{lemma}

A formal proof of Lemma \ref{lemma1} can be found in the appendix. However, note that the first two of the three statements of the lemma are not surprising: The first statement says that if you have two sets $A$ and $B$ which both contain $a$ and have the same size, then you need to assign the letter $a$ to equally many leaves in order to receive $A$ or $B$ as an MP root state estimate. This makes sense because the exact names of the states other than $a$ in a set do not matter; what matters is rather the number of other states in the respective set. The second statement says that the fully bifurcating tree of height 0, i.e. a tree that consists only of one node, requires this one node to be assigned $a$ if the root state estimate shall be $\{a\}$ -- this is clear as the root and the only leaf coincide in this case. The third statement is maybe the most interesting and least obvious one of the three statements. It says that if the root estimate shall be an unambiguous $\{a\}$, any fully bifurcating tree of height at least 1 requires at least two leafs that are assigned $a$.

Next we state a theorem which might seem obvious at first, too, but which is really crucial: If  you have two sets $A, B \subseteq R$, which both contain $a$, and if you want to obtain the {\it smaller} set as an MP root state estimate, {\it more} leaves have to be assigned $a$. 
\begin{theorem}\label{1}
Let $k\geq 1$ and $A, B \subseteq R$ such that $|A| \geq |B|$ with $a \in A \cap B$ and $|R|=r \geq 2$. Then, we have: $f_{k,r}^A \leq f_{k,r}^B$. 
\end{theorem} 
Note that, again, the statement of Theorem \ref{1} is intuitively clear, because, roughly speaking, the smaller the set of the MP root state estimate is, the more Maximum Parsimony  is convinced that $a$ is a good root state candidate. For instance, if $a$ is one state of only a few possible root states (i.e. many other states can be discarded as possible root state candidates), this outcome might be regarded as more relevant than if $a$ is one state of many possible states (i.e. only a few other states can be discarded). In order to achieve this higher degree of certainty, more leaves with $a$ are needed than if a smaller degree of certainty is sufficient.

We will now give a formal proof of Theorem \ref{1}, as not only the theorem itself, but also some aspects of the proof, are needed to establish subsequent results.

\begin{proof}
If $|A|=|B|$, then $f_{k,r}^A = f_{k,r}^B$ (see Lemma \ref{lemma1} (i)), so there remains nothing to show. Now assume $|A|>|B|$.  Note that without loss of generality, we may assume $B  \subset A$. This is due to the fact that if there were states in $B$ that are not contained in $A$, i.e. if $B \setminus (A \cap B)\neq \emptyset$, then we can re-name these states using states of set $A$ to get a set $B'$ with $|B'|=|B|$ and $a\in B'$. Then, by Lemma \ref{lemma1} (i), $f_{k,r}^B= f_{k,r}^{B'}$, but also $B' \subset A$. So we now assume that this re-naming of states in $B$ has already taken place, so without loss of generality, $B \subset A$. 

We proceed by induction on $k$. \\
For $k=1$ we have $|A|\leq 2^k=2$, because the number of states that can occur in any MP root state estimate is clearly limited by the number of leaves (by Definition \ref{parsimonyoperation} of the parsimony operation), as no states that occur in none of the leaves can ever enter the root state estimate. Moreover, it can easily be seen that $f_{1,r}^{\{a\}}=2$ and $f_{1,r}^{\{a,b\}}=1$ hold for all $r\geq 2$ (see Figure \ref{rschlange}) and therefore Theorem \ref{1} is true for $k=1$.
\begin{figure}[H]
\centering
\includegraphics[scale=0.5]{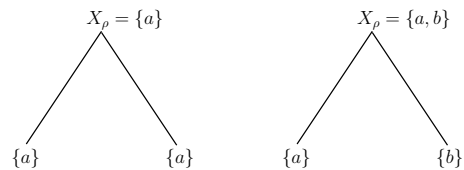}
\caption{Fully bifurcating trees of height 1 with $X_{\rho}=\{a\}$ and $X_{\rho}=\{a,b\}$ (w.l.o.g. $A_2=\{a,b\}$). To obtain $X_{\rho}=\{a\}$ both leaves have to be assigned $a$. To obtain $X_{\rho}=\{a,b\}$, one leaf has to be assigned $a$ and the other one $b$.}
\label{rschlange}
\end{figure}
Now suppose Theorem \ref{1} is true for $k$. Then we consider $k+1$ and use Equation \eqref{formel} to describe $f_{k+1,r}^A$ as follows:
\begin{align}
f_{k+1,r}^A &= \min
\begin{cases}
f_{k,r}^{A \cup S}+f_{k,r}^{A \cup \widehat{S}}: &S \cap \widehat{S}= \emptyset, S \cap A = \emptyset, \widehat{S} \cap A = \emptyset \\
&\text{and } |S| + |\widehat{S}| \leq r- |A|. \\
\\
f_{k,r}^{C_1}+f_{k,r}^{C_2}: &C_1, C_2 \neq \emptyset, C_1 \cap C_2 = \emptyset \text{ and } C_1 \cup C_2 =A\\
&\text{w.l.o.g. } a \in C_1. \\ 
\end{cases} \nonumber \\
&\qquad \qquad \text{(where the minimum is taken over all  $C_1, C_2, S, \widehat{S}$, respectively)} \nonumber \\
&= \min
\begin{cases}
f_{k,r}^{A \cup S}+f_{k,r}^{A \cup \widehat{S}}: &S \cap \widehat{S}= \emptyset, S \cap A = \emptyset, \widehat{S} \cap A = \emptyset \\
&\text{and } |S| + |\widehat{S}| \leq r- |A|. \\
\\
f_{k,r}^{C_1}: &a \in C_1, C_1 \subset A \text{ and } 1 \leq |C_1| \leq |A| -1.  \\
\end{cases} \nonumber \\
&\qquad \qquad \text{(since } a \not\in C_2 \text{ and therefore } f_k^{C_2}=0 \text{)} \nonumber \\
&= \min
\begin{cases}
f_{k,r}^{A \cup S}+f_{k,r}^{A \cup \widehat{S}}: &S \cap \widehat{S}= \emptyset, S \cap A = \emptyset, \widehat{S} \cap A = \emptyset \\
&\text{and } |S| + |\widehat{S}| \leq r- |A|. \\
\\
f_{k,r}^{C_1}: &a \in C_1, C_1 \subset A \text{ and } |C_1| = |A| -1.  \\
\end{cases} \label{fkrA}\\
&\qquad \qquad \text{(as  Theorem \ref{1} holds for $k$)} \nonumber. 
\end{align}
Analogously, we derive 
\begin{align*}
f_{k+1,r}^B
&= \min
\begin{cases}
f_{k,r}^{B \cup T}+f_{k,r}^{B \cup \widehat{T}} &T \cap \widehat{T}= \emptyset, T \cap B = \emptyset, \widehat{T} \cap B = \emptyset \\ 
&\text{and } |T| + |\widehat{T}| = r- |B|. \\
\\
f_{k,r}^{B_1} &a \in B_1, B_1 \subset B \text{ and } |B_1| = |B| -1.  \\
\end{cases} \nonumber \\
&\qquad \qquad \text{(where the minimum is taken over  $B_1, T, \widehat{T}$)} \nonumber 
\end{align*}

So for $f_{k+1,r}^{B}$, there are basically two cases, depending on whether this value is obtained by its two maximal pending subtrees in the standard decomposition by an intersection or a union. We consider these two cases for $f_{k+1,r}^B$ separately. \\
\textbf{$1^{\text{st}}$ case:} $f_{k+1,r}^B = f_{k,r}^{B_1}$ 
with $a \in B_1$, $B_1 \subset B$ and $|B_1| = |B| -1$. \\ Furthermore, we have that $A \subseteq R$ with $B \subset A$. Moreover, let $C_1 \subset A$ with $a \in C_1$ and $|C_1| = |A| -1$. Note that then $|B_1|<|C_1|$.  We conclude: 
\begin{align*}
f_{k+1,r}^B &= f_{k,r}^{B_1} \\
&\geq f_{k,r}^{C_1} &&\text{as Theorem \ref{1} holds for $k$}  \\
&\geq f_{k+1,r}^A &&\text{by  \eqref{fkrA}. } 
\end{align*}
\textbf{$2^{\text{nd}}$ case:} $f_{k+1,r}^B = f_{k,r}^{B \cup T}+f_{k,r}^{B \cup \widehat{T}}$
with $T \cap \widehat{T}= \emptyset$, $T \cap B = \emptyset$, $\widehat{T} \cap B = \emptyset$ and $|T| + |\widehat{T}| = r- |B|$. \\
\noindent
Moreover, we have $A \subseteq R$ and $B \subset A$. Now, as we have $B \subset A$, we can conclude $B \cup T \subseteq A \cup T \text{ for all } T
\text{ and }
B \cup \widehat{T} \subseteq A \cup \widehat{T} \text{ for all } \widehat{T}.$
Hence
\begin{align*}
f_{k+1,r}^B &= f_{k,r}^{B \cup T}+f_{k,r}^{B \cup \widehat{T}} \\
&\geq f_{k,r}^{A \cup T}+f_{k,r}^{A \cup \widehat{T}} &&\text{as Theorem \ref{1} holds for $k$}  \\
&\geq f_{k+1,r}^A &&\text{by  \eqref{fkrA}. }
\end{align*}
So in both cases, we conclude $f_{k+1,r}^A \leq f_{k+1,r}^B$, which completes the proof.
\qed
\end{proof}
The proof of Theorem \ref{1} (particularly Equation \eqref{fkrA}) helps us directly to gain more insight into $f_{k+1,r}^{A_i}$, which we will summarize in the following corollary.

\begin{col}\label{col1}
The following recursions hold: 
\begin{align*}
&f_{k+1,r}^{A_r}= f_{k+1,r}^R = \min\{f_{k,r}^{A_{r-1}}, 2 \cdot f_{k,r}^{A_r} \}, \\ 
&f_{k+1,r}^{A_{r-1}} = \min\{f_{k,r}^{A_{r-2}}, f_{k,r}^{A_{r-1}}+f_{k,r}^{{A_{r-1}} \cup D_1} \}, \\
&f_{k+1,r}^{A_{r-2}} = \min\{f_{k,r}^{A_{r-3}}, f_{k,r}^{A_{r-2}}+f_{k,r}^{{A_{r-2}} \cup D_2}, 2 \cdot f_{k,r}^{{A_{r-2}} \cup D_1} \}, \\
&\vdots \\
&f_{k+1,r}^{A_1}=f_{k+1} = \min\{f_{k,r}^{A_{1}}+f_{k,r}^{A_{1} \cup D_{r-1}}, f_{k,r}^{A_{1} \cup D_1}+f_{k,r}^{{A_{1}} \cup D_{r-2}}, f_{k,r}^{{A_{1}} \cup D_2}+f_{k,r}^{{A_{1}} \cup D_{r-3}}, \dots, F^{A_1} \}, \\
&\qquad \quad \text{where } F^{A_1} \coloneqq \begin{cases}
f_{k,r}^{A_1 \cup D_{p-1}}+f_{k,r}^{A_1 \cup D_{p}} &\text{if } r=2p \\
2 \cdot f_{k,r}^{A_1 \cup D_{p-1}}  &\text{if } r=2p-1
\end{cases} 
.
\end{align*}
\end{col}

As Corollary \ref{col1} is a direct consequence of \eqref{fkrA}, we omit the formal proof. Instead, we turn our attention to another rather intuitive property of $f_{k,r}^A$: namely, that as $k$ (and thus the tree under consideration) grows, the number of leaves needed for an MP root state estimate of $A$ cannot decrease. 

\begin{theorem}\label{2}For all regions where $f_{k,r}^A$ and  $f_{k+1,r}^A$ are defined, the following inequality holds: $
f_{k,r}^A \leq f_{k+1,r}^A.$ That is, $f_{k,r}^A$ is monotonically increasing in $k$ for all $A \subseteq R$.
\end{theorem}

The proof of Theorem \ref{2} is presented in the appendix. 

\par

Considering \eqref{fkrA}, it is obvious that in order to calculate $f_ {k+1,r}^{A}$ recursively, one has to consider two cases, which is due to the fact that the MP root state estimate for $T_{k+1}$ can come from a union or an intersection at the level of the two maximum pending subtrees of the standard decomposition. However, as the following crucial theorem shows, in the case where $2\leq |A| \leq r$, the recursion can be drastically simplified. 

\begin{theorem}\label{3} 
Let $i=2, \dots, r$. Then for all $k+1 \geq i-1$ for which $f_{k+1,r}^{A_i}$,   $f_{k,r}^{A_{i-1}}$ and $f_{k-(i-2),r}$ are defined, we have:
$$f_{k+1,r}^{A_i} = f_{k,r}^{A_{i-1}}= f_{k-(i-2),r}. $$
\end{theorem}

Note that Theorem \ref{3} is not only crucial in the sense that it allows us to omit a distinction between union and intersection -- it even allows for a reduction of a root state estimate $A$ with $|A|\geq 2$ to the unambiguous root state estimate $\{a\}$ but for a smaller tree. 
The proof of this theorem can be found in the appendix, but in brief the idea of the first equation stated in Theorem \ref{3} can be summarized like this: root state estimate $A_i$ can be achieved by considering the standard decomposition of $T_{k+1}$ into the two subtrees $T_k$, one of which gets assigned $A_{i-1}$ as a root state estimate, and the other one gets assigned a set, say $C$, with $|C|=1$ and with $a\not\in C$ and $A_{i-1}\cap C=\emptyset$. This way, the root estimate of $T_{k+1}$ will be $A_{i-1} \cup C = A_i$. Note that in order to achieve root state estimate $C$ for $T_k$, no leaf in $T_k$ needs to be labelled $a$ (e.g. all leafs in this tree can be labelled with the unique state in $C$, which is not $a$). So $f^C_{k,r}=0$ and thus $f_{k+1,r}^{A_i} = f_{k,r}^{A_{i-1}}+f^C_{k,r}=f_{k,r}^{A_{i-1}}$. The second equation of Theorem \ref{3} follows recursively.

\noindent
Note that by Theorem \ref{3} and Equation \eqref{fkrA}  we can also express $f_{k+1,r}$ in terms of  $f_{k,r}$. In particular, we have
\begin{align}
f_{k+1,r}&= \min\{f_{k,r}^{A_1}+f_{k,r}^{A_r},f_{k,r}^{A_2}+f_{k,r}^{A_{r-1}},f_{k,r}^{A_3}+f_{k,r}^{A_{r-2}}
,\dots, F_1 \} \label{f1}\\ 
&\quad \qquad \text{where } F_1 \coloneqq \begin{cases}
f_{k,r}^{A_p}+f_{k,r}^{A_{p+1}} &\text{if } r=2p \\
2 \cdot f_{k,r}^{A_{p}}  &\text{if } r=2p-1
\end{cases} \nonumber\\
&= \min\{f_{k,r}+f_{k-r+1,r}, f_{k-1,r}+f_{k-r+2,r}, f_{k-2,r}+f_{k-r+3,r}, \dots, F_2 \}, \label{f2}\\
&\quad \qquad \text{where } F_2 \coloneqq \begin{cases}
f_{k-p+1,r}+f_{k-p,r} &\text{if } r=2p \\
2 \cdot f_{k-p+1,r}  &\text{if } r=2p-1
\end{cases} \nonumber. 
\end{align}
\noindent
So in order for the MP root state estimate to be $\{a\}$, we still have to find the minimum over various cases. Note that the case $f_{k+1,r}=F_2$
corresponds to Conjecture \ref{conjecture}. This means that Conjecture \ref{conjecture} basically states that the other options can be discarded as they never give the minimum. As we have already seen in Observation \ref{obs}, this is not generally correct, though.

Before we can finally prove Main Theorem \ref{maintheo3}, we need one more useful property of function $f_{k,r}^A$. In particular, we find that if we reduce the number of character states that are available, the number of leafs that need to be in state $a$ in order for Maximum Parsimony to give a root state estimate of $A_i$ (where $a\in A_i$  and $|A_i|=i$ as before) cannot decrease. Again, this is intuitively clear, because we have already seen in the construction presented in the proof of Observation \ref{obs} that we could reduce the number of leafs in state $a$ by making use of as many alternative character states as possible.

\begin{theorem} \label{4}
Let $R$ be a finite set of character states with $|R|=r \geq 2$ and $a \in R$, let $2\leq \tilde{r} \leq r$ and $a \in A_i \subseteq R$ with $i \in \{1,\dots, \tilde{r}\}$. Then, for all regions where $f_{k+1,r}^{A_i}$ and  $f_{k+1,\tilde{r}}^{A_i}$ are defined, we have: 
$$f_{k+1,r}^{A_i} \leq f_{k+1,\tilde{r}}^{A_i}.$$
\end{theorem}

A formal proof of Theorem \ref{4} can again be found in the appendix. This concludes our summary of basic properties of $f_{k,r}^A$, which we will use subsequently to prove our main results. We first focus our attention again on the case where $r$ is odd; in particular, on the special case where $r=3$.

\subsection{Results for 3-state characters} \label{3-state}
We now consider the case where $|R|=r=3$. By Theorem \ref{1} and Theorem \ref{3}, we already have for $r=3$ that
\begin{align}
&f_{k+1,3}^{A_3}=f_{k+1,3}^R= f_{k,3}^{A_2}, \label{w1} \\
&f_{k+1,3}^{A_2}= f_{k,3}, \label{w2}\\
&f_{k+1,3}^{A_1}=f_{k+1,3} = \min\{ 2 \cdot f_{k,3}^{A_2}, f_{k,3}+f_{k,3}^{R} \}  \nonumber \\
&\phantom{ppppppppppppp} =\min\{ 2 \cdot f_{k-1,3}, f_{k,3}+f_{k-2,3} \} \label{w3}
.
\end{align}
We are now finally in the position to prove Main Theorem \ref{maintheo3}.

\begin{proof}[Main Theorem \ref{maintheo3}]
First consider the statements for $f_{2,3}$ and $f_{3,3}$. As the construction presented in Figure \ref{t_3} shows, $f_{3,3}\leq 3$. Moreover, the left-hand side subtree in Figure  \ref{t_3} also shows that $f_{2,3}\leq 2$. An exhaustive search through all possible characters on $2^2=4$ and $2^3=8$ leafs, respectively, shows that in both cases, equality holds (calculations not shown). \\
So now we turn our attention to the case $k \geq 3$. Without loss of generality, we can assume that $k$ odd and prove Theorem \ref{maintheo3} separately for $k$ and $k+1$, where the latter then covers the case where the tree height is even. 

We first show the following two statements by induction on $k$:
\begin{align}
& 2 \cdot f_{k-1,3} < f_{k,3} + f_{k-2,3}, \label{for31} \\
& 2 \cdot f_{k,3} = f_{k+1,3} + f_{k-1,3} \label{for32}.
\end{align}
For $k = 3$ we have $f_{1,3}=2$ (see Figure \ref{rschlange}), $f_{2,3}=2$ and $f_{3,3}=3$ (initial conditions as explained above) and, by exhaustive search, $f_{4,3}=4$ (reached by character $abacbbbbabaccccc$ on $2^4$ leafs; calculation not shown). This leads to
\begin{align*}
&2 \cdot f_{2,3} =  4 < 5 = f_{3,3}+f_{1,3}, \\
&2 \cdot f_{3,3} =  6 = f_{4,3}+f_{2,3},
\end{align*}
which is the base case of the induction. Now suppose \eqref{for31} and \eqref{for32} are true for all $h$ with $3 \leq h \leq k$ and $h$ odd. Then with \eqref{w3}, $f_{k+1,3}$ and $f_{k+2,3}$ become 
\begin{align}
&f_{k+1,3} = 2 \cdot f_{k-1,3}, \label{for33} \\
&f_{k+2,3} = 2 \cdot f_{k,3} = f_{k+1,3}+f_{k-1,3} \label{for34},
\end{align}
where the second equation in \eqref{for34} is due to \eqref{for32}. 

Now we show that \eqref{for31} and \eqref{for32} are also true for $k+2$.
\begin{align*}
2 \cdot f_{(k+2)-1,3} &= 2 \cdot f_{k+1,3} \\
&= 2 \cdot 2 \cdot f_{k-1,3} &&\text{by } \eqref{for33} \\
&< 2 \cdot (f_{k,3}+f_{k-2,3}) &&\text{by } \eqref{for31} \\
&= 2 \cdot f_{k,3} + 2 \cdot f_{k-2,3} \\
&= f_{k+2,3} + f_{k,3} &&\text{by } \eqref{for34}
. 
\end{align*}
Therefore \eqref{for31} holds for all $k\geq 3$ and $k$ odd.
\begin{align*}
2 \cdot f_{k+2,3} &= 2 \cdot (f_{k+1,3}+f_{k-1,3}) &&\text{by } \eqref{for34} \\ 
&= 2 \cdot f_{k+1,3} + 2 \cdot f_{k-1,3} \\
&= f_{k+3,3} + f_{k+1,3} &&\text{by } \eqref{for31}
.
\end{align*}
Hence also \eqref{for32} holds for all $k\geq 3$ and $k$ odd, so that in total, \eqref{for33} and \eqref{for34} are also true for all $k \geq 3$ and $k$ odd. Thus we have that $f_{k+1,3}=2 \cdot f_{k-1,3}$ holds for all $k \geq 3$. So together with Theorem \ref{3}, we have $2 \cdot f_{k-1,3} = 2 \cdot f_{k,3}^{A_2}$, which completes the proof.
\qed 
\end{proof}

Recall that Main Theorem \ref{maintheo3} basically gives a characterization of the cases in which Conjecture \ref{conjecture} holds in the case where $r=3$; namely by choosing appropriate initial conditions in order to avoid the scenario constructed in the proof of Observation \ref{obs}. However, this result cannot easily be extended to the general case where $r$ is odd. We will discuss this further in Section \ref{discussion}.

We now turn our attention to the case where $r$ is even, which will turn out to be less involved as Conjecture \ref{conjecture} can be shown to hold in this case.

\subsection{Results for all even numbers of character states}

We are finally in a position to state our second main result. 

\begin{maintheorem} \label{gerade}
Let $T_k$ be a fully bifurcating tree of height $k$ and let $R$ be a set of character states with $|R|=r=2p$ and $p \in \mathbb{N}_{\geq 1}$. Then for $k \geq r$ and all regions where $f_{k+1,r}$, $f_{k,r}^{A_p}$,  $f_{k,r}^{A_{p+1}}$, $f_{k-p+1,r}$  and $f_{k-p,r}$ are defined,  we have:
$$
f_{k+1,r} = f_{k-p+1,r} + f_{k-p,r}.  
$$ 
\end{maintheorem}

Note that Theorem \ref{gerade} states that Conjecture \ref{conjecture} holds in the case where $r$ is even and $k\geq r$, and this statement does not depend on a specific choice of initial conditions. Before we can present the proof of Main Theorem \ref{gerade}, we need an additional lemma, which helps to cover the case $k<r$, which is not considered in the theorem. Recall that we already have $f_{0,r}=1$  and $f_{1,r}=2$ for all $r$ (by Lemma \ref{lemma1} (ii) and Figure \ref{rschlange}). The following lemma investigates the case $k < r$ further.

\begin{lemma} \label{lemmar2p}
For $r=2p$ with $p\in \mathbb{N}_{\geq 1}$, we have:  
\begin{enumerate}
\item $f_{k,r}=2$ for all $k=1, \dots, p$
\item $f_{p+1,r}=3$
\item $f_{k,r} \leq 4$ for all $p+1<k \leq r$.
\end{enumerate}
\end{lemma}

The proof of this lemma can be found in the appendix. Note that the bound stated in Lemma \ref{lemmar2p} is actually tight, as it can be shown that while $f_{5,6}=f_{6,8}=3$, we have $f_{7,8}=4$ (calculations not shown). Whether $f_{k,r}=3$ is possible for $r=2p \geq 10$ and $p+1<k \leq r$ has not yet been investigated.

We now state the proof of Main Theorem \ref{gerade}.

\begin{proof}[Main Theorem \ref{gerade}] Note that by the second equation of Theorem \ref{3}, we have $f_{k-p+1,r}+f_{k-p,r} =f_{k,r}^{A_p} + f_{k,r}^{A_{p+1}}$. We now show that for 
all $r=2p$ with $p \in \NN_{\geq1}$ we have 
\begin{align}
f_{k,r}^{A_p} + f_{k,r}^{A_{p+1}} \leq f_{k,r}^{A_i} + f_{k,r}^{A_{j}} \label{geradeib}
\end{align}
for all $1 \leq i \leq j \leq r$ and $i+j = r+1$. Together with \eqref{f1} this will prove Main Theorem \ref{gerade}.
The proof is by induction on $k$.
For $k=r$ we need to show that 
\begin{align}\label{toshow}
f_{r,r}^{A_p}+f_{r,r}^{A_{p+1}} \leq \min\{f_{r,r}^{A_1}+f_{r,r}^{A_r},f_{r,r}^{A_2}+f_{r,r}^{A_{r-1}},f_{r,r}^{A_3}+f_{r,r}^{A_{r-2}}
,\ldots, f_{r,r}^{A_{p-1}}+f_{r,r}^{A_{p+2}} \},
\end{align}
which is, by Theorem \ref{3}, equivalent to showing 
\begin{align*}
f_{r-p+1,r}+f_{r-p,r} \leq \min\{ f_{r,r}+f_{r-r+1,r},f_{r-1,r}+f_{r-r+2,r},\ldots, f_{r-p+2,r}+f_{r-p-1,r}\}.
\end{align*}
Since $r=2p$, this inequality can be written as
\begin{align}
f_{p+1,r}+f_{p,r} \leq \min\{f_{r,r}+f_{1,r},f_{r-1,r}+f_{2,r},\ldots, f_{p+2,r}+f_{p-1,r}\}. \label{rk}
\end{align}
Moreover, by Lemma \ref{lemmar2p} we know that $f_{1,r}=f_{2,r}=\dots=f_{p,r}=2$ and thus \eqref{rk} results in 
\begin{align}\label{simplification}
f_{p+1,r} \leq \min\{f_{r,r},f_{r-1,r},\ldots, f_{p+2,r}\}.
\end{align}
By Theorem \ref{2} we have 
$
f_{p+1,r} \leq f_{p+2,r} \leq \dots \leq f_{r-2,r} \leq f_{r-1,r} \leq f_{r,r}.
$
So altogether, \eqref{simplification} holds and thus also \eqref{toshow}. This complese the base case of the induction. \\
Now assume that \eqref{geradeib} holds for $k$, then with \eqref{f1} and Theorem \ref{3},  we get 
\begin{align}\label{gerade1}
f_{k+1,r} = f_{k,r}^{A_p} + f_{k,r}^{A_{p+1}}  = f_{k-p+1,r}+f_{k-p,r} . \end{align}

Now we show that \eqref{geradeib} also holds for $k+1$.
\begin{align*}
f_{k+1,r}^{A_p} + f_{k+1,r}^{A_{p+1}} &= f_{k-(p-2),r} + f_{k-(p+1-2),r} &&\text{by Theorem } \ref{3} \\
&= f_{k-p+2,r} + f_{k-p+1,r} \\
&= f_{k-p+1,r}^{A_p} + f_{k-p+1,r}^{A_{p+1}} + f_{k-p,r}^{A_p} + f_{k-p,r}^{A_{p+1}} &&\text{by } \eqref{gerade1} \\
&\leq f_{k-p+1,r}^{A_i} + f_{k-p+1,r}^{A_{j}} + f_{k-p,r}^{A_i} + f_{k-p,r}^{A_{j}} &&\text{by } \eqref{geradeib}\\
&\quad \text{for all } i,j: 1 \leq i \leq j \leq r \text{ and } i+j=r+1 \\
&= f_{k-p-(i-2),r} + f_{k-p-(j-2),r} + f_{k-p-(i-1),r} + f_{k-p-(j-1),r} &&\text{by Theorem } \ref{3} \\
&= f_{k-(p-1)-i+1,r} + f_{k-(p-1)-j+1,r} + f_{k-(p-1)-i,r} + f_{k-(p-1)-j,r} \\
&= f_{k-i+1,r}^{A_p} + f_{k-j+1,r}^{A_{p}} + f_{k-i+1,r}^{A_{p+1}} + f_{k-j+1,r}^{A_{p+1}} &&\text{by Theorem } \ref{3} \\
&= f_{k-i+1,r}^{A_p} + f_{k-i+1,r}^{A_{p+1}} + f_{k-j+1,r}^{A_{p}} + f_{k-j+1,r}^{A_{p+1}} \\
&= f_{k+1-i+1,r} + f_{k+1-j+1,r} &&\text{by } \eqref{gerade1} \\
&= f_{k-(i-2),r} + f_{k-(j-2),r} \\
&= f_{k+1,r}^{A_i} + f_{k+1,r}^{A_j} &&\text{by Theorem } \ref{3}
.
\end{align*}
Therefore \eqref{geradeib}, holds for all $r=2p$ with $p \in \NN^+$. Applying \eqref{geradeib} to \eqref{f1} and \eqref{f2} completes the proof.
\qed
\end{proof}

\section{Conclusion and Discussion}\label{discussion}
We considered the question how many leafs of a fully bifurcating tree must be labelled $a$ in order for the Maximum Parsimony root state estimate to be $a$. In this regard, we could prove a conjecture by Steel and Charleston for the case where the number $r$ of character states is even \cite{paper}. This is an interesting result because it not only proves the conjecture, but also generalizes the findings presented in \cite{paper}, where the case $r=2$ was investigated. We proved that the case where $r$ is even is simple in the sense that the conjecture holds no matter where the recursion is started.

On the other hand, we also showed that when $r$ is odd, counterexamples to the conjecture can be constructed under certain conditions. However, we were able to prove the conjecture for $r=3$ for a careful choice of initial conditions. The general case where $r$ is odd is a topic of ongoing research. 

The conjecture by Steel and Charleston and thus also our results are relevant as the fully bifurcating tree is a `worst case scenario' for unambiguous Maximum Parsimony root state estimation: As all leafs are equally far away from the root (in the sense of the number of edges that separates them from the root), there is no leaf that has a higher impact on the root state estimate than others. We showed in Section \ref{intro} that the so-called caterpillar tree is the other extreme: Here, no matter how many leafs we label with $b$, two (particular) leafs labelled $a$ will be sufficient to get an unambiguous MP root state estimate of $\{a\}$, whereas for the fully balanced tree, the number depends strongly on the size of the tree (in terms of the height and thus the number of leafs). Recall that in the literature, the fully bifurcating tree is often referred to as `balanced' as opposed to the caterpillar tree, which is often thought of as `unbalanced' \cite{liebscher}. Trees with $2^k$ leafs (for some $k\geq 1$ which are `in-between' these two trees in terms of balance will also need at least 2 leafs to be in state $a$, like the caterpillar tree, and at most $f_{k,r}$ leafs in state $a$ like the fully bifurcating tree, in order for MP to give an unambiguous root state estimate of $\{a\}$. Thus, our results present an upper bound for all trees of size $2^k$. However, the exact numbers for other tree shapes as well as other tree sizes are still to be investigated. 

\begin{acknowledgements}
We thank Mike Steel for bringing this topic to our attention. The first author also thanks the Ernst-Moritz-Arndt-University Greifswald for the Landesgraduiertenf\"orderung studentship, under which this work was conducted.
\end{acknowledgements}

\section{Appendix}

\begin{proof}[Proof of Lemma \ref{lemma1}]
(i) Let $A$, $B$ such that $a \in A \cap B$ and $|A| = |B|$. Then $A$ could be transformed into $B$ by renaming all states not element of $A \cap B$. Then $A=B$, and this yields $f_{k,r}^A=f_{k,r}^B$. 
\\
(ii) Let $k=0$. In this case our tree consists of one leaf, which is at the same time the root. This vertex have to be assigned $a$ to obtain $X_\rho=\{a\}$. Hence $f_{0,r}=1$ for all $r$.
\\
(iii) Let $k \geq 1$. Then
\begin{align*}
f_{k,r} &= \min\{f_{k-1,r}^B+f_{k-1,r}^C: B*C=\{a\} \} \\
&= \min\{f_{k-1,r}^B+f_{k-1,r}^C: B \cap C=\{a\} \}, \\
&\text{since } B \cup C \text{ would not result in } X_\rho=\{a\} \text{ as $B,C\neq \emptyset$ } .
\end{align*}
Therefore, $a \in B$ and $a \in C$ which results in $f_{k,r} \geq 2$.
\qed
\end{proof}

\begin{proof}[Theorem \ref{2}]
We have $f_{0,r}^{\{a\}}=1$ (by Lemma \ref{lemma1} (ii)) and $f_{1,r}^{\{a\}}=2$ (see Figure \ref{rschlange}) for all $r$. Therefore, Theorem \ref{2} is true for $k=0$.
We now prove Theorem \ref{2} for $k>0$ by contradiction. So assume  
there is a $\widehat{k} \text{ such that for some set } A \text{ we have } f_{\widehat{k}+1,r}^A < f_{\widehat{k},r}^A. $
Choose $k$ to be the smallest value of $\widehat{k}$ with this property, i.e. we have $f_{k+1,r}^A < f_{k,r}^A $ for some set $A$, which we fix. Moreover, for all $\widetilde{k} < k$ and all sets $\widetilde{A}$ we have $f_{\widetilde{k},r}^{\widetilde{A}} \geq f_{\widetilde{k}-1,r}^{\widetilde{A}}.$
Now let $A_1, S, \widehat{S}$ and $f_{k+1,r}^A$ be as in the proof of Theorem \ref{1}. We now consider both cases for $f_{k+1,r}^A$ for the fixed values $k$ and $A$ as chosen above.\\
\textbf{$1^{\text{st}}$ case:} $f_{k+1,r}^A=f_{k,r}^{A_1}$
with $a \in A_1$, $A_1 \subset A$ and $|A_1| = |A| - 1$.
Then $f_{k,r}^{A_1} = f_{k+1,r}^A < f_{k,r}^A $ by the assumption on $k$.
This leads to $f_{k,r}^{A_1} < f_{k,r}^A$, which contradicts Theorem \ref{1}, because $|A_1|=|A|-1 < |A|$.\\
\textbf{$2^{\text{nd}}$ case:} $f_{k+1,r}^A=f_{k,r}^{A \cup S}+f_{k,r}^{A \cup \widehat{S}}$
with $S \cap \widehat{S}= \emptyset$, $S \cap A = \emptyset$, $\widehat{S} \cap A = \emptyset$ and $|S| + |\widehat{S}| = r- |A|$.
Then
\begin{align*}
f_{k,r}^A &> f_{k+1,r}^A &&\text{by assumption on $k$} \\ 
&=f_{k,r}^{A \cup S}+f_{k,r}^{A \cup \widehat{S}} \\
&\geq f_{k-1,r}^{A \cup S}+f_{k-1,r}^{A \cup \widehat{S}} &&\text{by assumption on all $\widetilde{k}<k$, in particular $\widetilde{k}=k-1$} \\
&\geq f_{k,r}^A &&\text{by \eqref{fkrA}}.
\end{align*}
Thus $f_{k,r}^A > f_{k,r}^A$, which is a contradiction.\\ 
Therefore, both cases lead to contradictions, and thus, such a $k$ cannot exist. This completes the proof.
\qed
\end{proof}

\begin{proof}[Theorem \ref{3}]
We begin by proving
\begin{align}
f_{k,r}^{A_{i-1}} \leq f_{k,r}^{A_i \cup S}+f_{k,r}^{A_i \cup \widehat{S}} \label{theorem3}
\end{align}
for $2 \leq i \leq r$ for all $S, \widehat{S} \subseteq R:$ $S \cap \widehat{S}= \emptyset$, $S \cap A_i = \emptyset$, $\widehat{S} \cap A_i = \emptyset$ and $|S| + |\widehat{S}| = r- |A_i|$ by induction on $k$. 
For $k=1$ we have that $f_{1,r}^{A_1}=2$ and $f_{1,r}^{A_2}=1$ hold for all $r$ (see Figure \ref{rschlange}). Thus $2=f_{1,r}^{A_1} \leq f_{1,r}^{A_2}+f_{1,r}^{A_2}=1+1$ and therefore \eqref{theorem3} is true for $k=1$,  which completes the base case of the induction. Assume now that for all $i=2, \dots, r$, \eqref{theorem3} holds for $k$.  By \eqref{fkrA} and \eqref{theorem3}, for all $A_i$, $f_{k+1,r}^{A_i}$ becomes
\begin{align}
f_{k+1,r}^{A_i} &= \min
\begin{cases}
f_{k,r}^{A_i \cup S}+f_{k,r}^{A_i \cup \widehat{S}} &S \cap \widehat{S}= \emptyset, S \cap A_i = \emptyset, \widehat{S} \cap A_i = \emptyset \\ 
&\text{and } |S| + |\widehat{S}| = r- |A_i|. \\ \\
f_{k,r}^{A_{i-1}} &a \in A_{i-1}, A_{i-1} \subset A_i \text{ and } |A_{i-1}| = i-1.  \\
\end{cases} \label{3.0} \\
&\qquad \qquad \text{where the minimum is taken over all } A_{i-1}, S, \widehat{S} \nonumber \\
&= f_{k,r}^{A_{i-1}}. \label{3.1}
\end{align}

Then for all $2 \leq i \leq r$  and for all $S, \widehat{S} \subseteq R:$ $S \cap \widehat{S}= \emptyset$, $S \cap A_i = \emptyset$, $\widehat{S} \cap A_i = \emptyset$ and $|S| + |\widehat{S}| = r- |A_i|$, we can apply \eqref{3.1} to $A_i\cup S$ and $A_i\cup\widehat{S}$, respectively, and get:
\begin{align*}
f_{k+1,r}^{A_i \cup S}+f_{k+1,r}^{A_i \cup \widehat{S}} 
&= f_{k,r}^{A_{i-1} \cup S}+f_{k,r}^{A_{i-1} \cup \widehat{S}} &&\text{by } \eqref{3.1} \\
&\geq f_{k,r}^{A_{i-2}} &&\text{by } \eqref{theorem3} \\
&=f_{k+1,r}^{A_{i-1}} &&\text{by } \eqref{3.1}.
\end{align*}

Note that the first equation is true because as for $A_i\cap S=\emptyset$, we have $|A_{i-1}\cup S|=|A_{i}\cup S|-1$. The same holds for $\widehat{S}$.
So altogether, we now have that \eqref{theorem3} holds for $k+1$.\\
Applying \eqref{theorem3} to \eqref{3.0} gives $f_{k+1,r}^{A_i}=f_{k,r}^{A_{i-1}}$ for all $2 \leq i \leq r$. So now we have that \eqref{3.1} holds for all $k+1\geq i-1$, and thus we have $f_{k,r}^{A_{i-1}}=f_{k-1,r}^{A_{i-2}}= f_{k-2,r}^{A_{i-3}}= \ldots =f_{k-(i-2),r}^{A_{i-(i-1)}} =f_{k-(i-2),r}^{A_1}=f_{k-(i-2),r},$ which completes the proof.

\qed
\end{proof}

\begin{proof}[Theorem \ref{4}]
If $\tilde{r}=r$, we immediately have $f_{k+1,r}^{A_i}=f_{k+1,\tilde{r}}^{A_i}$, so there is nothing to show. So we now consider the case $\tilde{r}<r$ and prove the statement by induction on $k$.  
For $k=0$, we consider  $f_{1,r}^{A_i}$ and  $f_{1,\tilde{r}}^{A_i}$. These values are defined only if $|A_i|\leq 2^{0+1}=2$ and therefore $i=1$ or $i=2$. Thus we have to consider $f_{1,r}^{A_1},f_{1,\tilde{r}}^{A_1},f_{1,r}^{A_2}$ and $f_{1,\tilde{r}}^{A_2}$ for all $2 \leq \tilde{r} \leq r$. As can be seen in Figure \ref{rschlange}, we have for all $2 \leq \tilde{r} \leq r$:
$f_{1,r}^{A_1} =2= f_{1,\tilde{r}}^{A_1}$ and $f_{1,r}^{A_2} =1= f_{1,\tilde{r}}^{A_2}.$ This completes the base case of the induction for all $i\in \{1,\dots,\tilde{r}\}$. 

Now assume Theorem \ref{4} holds for $k$, and consider $k+1$. We start with the case $i\in \{2,\dots,\tilde{r}\}$ and find:
\begin{align*}
f_{k+2,\tilde{r}}^{A_i} &= f_{k+1,\tilde{r}}^{A_{i-1}} &&\text{by Theorem } \ref{3}\\
&\geq f_{k+1,r}^{A_{i-1}} &&\text{as Theorem } \ref{4} \text{ holds for } k \\
&=f_{k+2,r}^{A_i} &&\text{by Theorem } \ref{3}.
\end{align*}
So Theorem \ref{4} holds for $i\in \{2,\dots,\tilde{r}\}$. 
Now consider the case $i=1$:
\begin{align*}
f_{k+2,\tilde{r}}^{A_1} &= \min\{f_{k+1,\tilde{r}}^{A_1}+f_{k+1,\tilde{r}}^{A_{\tilde{r}}},
f_{k+1,\tilde{r}}^{A_2}+f_{k+1,\tilde{r}}^{A_{\tilde{r}-1}},f_{k+1,\tilde{r}}^{A_3}+f_{k+1,\tilde{r}}^{A_{\tilde{r}-2}},\dots, F_3 \} &&\text{by } \eqref{f1}\\
&\quad \qquad \text{where } F_3 \coloneqq \begin{cases}
f_{k+1,\tilde{r}}^{A_{\tilde{p}}}+f_{k+1,\tilde{r}}^{A_{\tilde{p}+1}} &\text{if } \tilde{r}=2 \tilde{p} \\
2 \cdot f_{k+1,\tilde{r}}^{A_{\tilde{p}}}  &\text{if } \tilde{r}=2 \tilde{p} -1
\end{cases} \\
&\geq \min\{f_{k+1,r}^{A_1}+f_{k+1,r}^{A_{\tilde{r}}},
f_{k+1,r}^{A_2}+f_{k+1,r}^{A_{\tilde{r}-1}},f_{k+1,r}^{A_3}+f_{k+1,r}^{A_{\tilde{r}-2}},\dots, F_4 \} &&\text{as Theorem } \ref{4} \text{ holds for } k\\
&\quad \qquad \text{where } F_4 \coloneqq \begin{cases}
f_{k+1,r}^{A_{\tilde{p}}}+f_{k+1,r}^{A_{\tilde{p}+1}} &\text{if } r=2p \\
2 \cdot f_{k+1,r}^{A_{\tilde{p}}}  &\text{if } r=2p-1
\end{cases} \\
&\geq \min\{f_{k+1,r}^{A_1}+f_{k+1,r}^{A_r},
f_{k+1,r}^{A_2}+f_{k+1,r}^{A_{r-1}},f_{k+1,r}^{A_3}+f_{k+1,r}^{A_{r-2}},\dots, F_4 \} &&\text{by Theorem } \ref{1}\\
&\quad \qquad \text{where } F_4 \coloneqq \begin{cases}
f_{k+1,r}^{A_{p}}+f_{k+1,r}^{A_{p+1}} &\text{if } r=2p \\
2 \cdot f_{k+1,r}^{A_{p}}  &\text{if } r=2p-1
\end{cases} \\
&=f_{k+2,r}^{A_1} &&\text{by } \eqref{f1}.
\end{align*}
Hence also in the case where $i=1$, Theorem \ref{4} holds. This completes the proof.
\qed
\end{proof}

\begin{proof}[Lemma \ref{lemmar2p}]
Let $r= 2p\geq 2$, i.e. $p\geq 1$. 
\begin{enumerate} \item We start with the case $k\leq p$. We have $f_{1,r}=2$ for all $r$ (see Figure \ref{rschlange}). By Theorem \ref{2} we know that $f_{k,r}$ is monotonically increasing in $k$, and thus $2=f_{1,r} \leq f_{p,r}$.
We now use the standard decomposition for $T_p$ in order to derive its two maximal pending rooted subtrees $T_{p-1}$ with roots $\rho_1$ and $\rho_2$. Using the same construction as in the proof of Observation \ref{obs}, which is depicted in Figure \ref{paperfig}, we can achieve $X_{\rho_1}=A_p=\{a,c_1, \dots, c_{p-1}\}$ and $X_{\rho_2}=A_p^{'}=\{a,c_p, \dots, c_{2p-2}\}$ by assigning $a$ to  one leaf in each subtree $T_{p-1}$, respectively, where $R=\{a,c_1,\dots,c_{2p-1}\}$, and such that all other subtrees use one state each which is unique to this subtree. Thus, the root of $T_p$ will have the MP root state estimate $A_p \cap A_p^{'}=\{a\}$. (Note that no leaf is assigned character state $c_{2p-1}$, so we do not even require all states. We will need this fact later.) So we conclude $f_{p,r}\leq 2$, which together with $f_{p,r}\geq 2$ as explained above completes $f_{p,r}=2$. So together with Theorem \ref{2} we achieve that $f_{k,r}=2$ for all $k=1,\dots,p$.
\item Now consider the case $k=p+1$. In this case we have with \eqref{f1}
\begin{align*}
f_{p+1,r}=f_{p+1,2p} &=\min\{f_{p,2p}+f_{p,2p}^{A_{2p}},f_{p,2p}^{A_{2}}+f_{p,2p}^{A_{2p-1}},\ldots, f_{p,2p}^{A_{p}}+f_{p,2p}^{A_{p+1}} \} \\
&=\min\{f_{p,2p}+f_{p,2p}^{A_{2p}},f_{p-1,2p}+f_{p,2p}^{A_{2p-1}},\ldots, f_{1,2p}+f_{0,2p} \} &&\text{by Theorem } \ref{3} \\
&=\min\{2+f_{p,2p}^{A_{2p}},2+f_{p,2p}^{A_{2p-1}},\ldots, 2+1 \} &&\text{by Lemma \ref{lemmar2p}, part 1 }  \\
&=3.\end{align*}
The latter equation is true because $1 \leq f_{p,2p}^{A_{2p}} \leq f_{p,2p}^{A_{2p-1}} \leq \dots \leq f_{p,2p}^{A_{p+1}}$ (at least one leaf hast to be labelled $a$ if $a$ shall appear in the MP root state estimate).

\item Now consider the case $k=r$. We can proceed as above and assign $X_{\rho_1}=A_p=\{a,c_1, \dots, c_{p-1}\}$ to the first of the two $T_{r-1}$ subtrees induced by the standard decomposition, and $X_{\rho_2}=A_{p+1}^{'}=\{a,c_p, \dots, c_{2p-1}\}$ to the other one, where again $R=\{a,c_1,\dots,c_{2p-1}\}$. Then by \eqref{f1} we can conclude $f_{r,r}\leq f_{r-1,r}^{A_p} + f_{r-1,r}^{A_{p+1}}$. Note that $f_{r-1,r}^{A_p} + f_{r-1,r}^{A_{p+1}}= f_{p,r}+f_{p-1,r}$ by Theorem~\ref{3}, so that $f_{r,r}\leq f_{p,r}+f_{p-1,r} = 2+2$, where the latter equation holds because of the first part of Lemma \ref{lemmar2p}. So altogether, $f_{r,r}\leq 4$. By the monotonicity of Theorem~\ref{2}, we obtain $f_{k,r} \leq 4$ for all $p+1 < k \leq r$.
\end{enumerate}
\qed
\end{proof}

\end{document}